\documentclass[10pt]{article}
\usepackage{amssymb}
\usepackage{amsmath}
\usepackage{amsthm}
\usepackage{latexsym}
\usepackage[dvips]{epsfig}
\usepackage{mathrsfs}
\usepackage{xcolor}
\usepackage[bookmarks,colorlinks=false]{hyperref}
\usepackage{comment}

\theoremstyle{plain}

\newtheorem{remark}{Remark}
\newtheorem{lemma}{Lemma}
\newtheorem{theorem}{Theorem}

\newtheorem{definition}{Definition}

\newcounter{mnotecount}[section]

\newcommand{\mnotex}[1]%{}
{\protect{\stepcounter{mnotecount}}$^{\mbox{\footnotesize $\bullet$\themnotecount}}$ 
\marginpar{%\color{red}%
\raggedright\tiny\em
$\!\!\!\!\!\!\,\bullet$\themnotecount: #1} }

\begin{document}

\title{\textbf{An ideal conformally
covariant characterization of the Kerr conformal structure}}

\author{ A. Garc\'\i a-Parrado \\
	Departamento de Matem\'aticas, Universidad de C\'ordoba,\\
	Campus de Rabanales, 14071, C\'ordoba, Spain\\
	e-mail: agparrado@uco.es  \\[2ex]
}
\maketitle

\begin{abstract}
We present an \emph{ideal, conformally covariant} characterization of the family of four
dimensional Lorentzian spacetimes
that are conformally related to the Kerr vacuum solution.
\end{abstract}


\section{Introduction}
Conformal geometry has played a very important role in general relativity (see \cite{JUANBOOK} for a complete review of its applications).
The main reason for this is that a \emph{conformal structure} encompasses naturally a family of spacetimes which share the same
\emph{causal structure}. A conformal structure is a pair formed by a smooth manifold and an equivalence class formed by Lorentzian metrics
that are conformally related to each other (see \cite{CURRYGOVER} for precise definitions). The simplest conformal structure is the set of
Lorentzian metrics conformally related to the flat Minkowski metric (conformally flat spacetimes). As is well-known, a Lorentzian metric is
conformally flat if and only if its Weyl tensor vanishes (in this paper we assume spacetimes of dimension 4).
The Weyl tensor is defined only in terms
of \emph{concomitants} of the metric tensor (a concomitant is a quantity defined in terms of the metric tensor, its inverse and the Levi-Civita connection of the metric tensor). Such a characterization will be termed as an \emph{ideal characterization}. The advantage of ideal characterizations
is that they can be tested \emph{algorithmically} for any given Lorentzian metric. For example, if we wish to know whether a Lorentzian metric is
conformally flat we only need to compute its Weyl tensor and check wether it vanishes or not.

One can seek ideal characterizations for other conformal structures similar to the one des\-cribed above for the
Minkowski's conformal structure. In general relativity the vacuum solutions play a fundamental role and they are mathematically characterized as
\emph{Einstein spaces}. This is a characterization that is ideal in the sense described in the previous paragraph
but unfortunately it is not \emph{conformally covariant}. This means that the set of vacuum solutions in general relativity
does not define a conformal structure. However, if we are given a Lorentzian metric we can ask ourselves whether such a metric
is conformal to a vacuum solution. Ideal characterizations that answer the previous question have been studied in
\cite{KoNewTod85,Listing2001,Edgar2004,GOVER2006450}.

The aim of this paper is to present an ideal characterization of the family of Lorentzian metrics that are conformally
related to the Kerr solution \cite{KERR-METRIC}. This characterization is presented as a set of metric concomitants in
Theorem \ref{theo:kerr-conformal} which is the main result of this work. The relevance of this ideal characterization
is that it does not require the explicit computation of the conformal factor that relates the Lorentzian metric we are testing
with the Kerr metric. Moreover, we prove that our ideal characterization of the Kerr conformal family is also \emph{conformally covariant}.

This paper is organized as follows: in section \ref{sec:conformal-eqs} we summarize the notions of
conformal geometry that are needed. Section \ref{sec:conformal-einstein} presents an ideal, conformally covariant characterization of
spacetimes that are conformal to vacuum solutions of the Einstein field equations (Theorem \ref{theo:ideal-einstein} taken
from \cite{GOVER2006450}). In section \ref{sec:type-D} we review the ideal characterization of the Kerr solution
that was found in \cite{FERSAEZKERR}. Finally section \ref{sec:kerr-conformal} is devoted to our main result given in Theorem
\ref{theo:kerr-conformal}.

All the tensor computations in this paper have been carried out with the system {\em xAct} \cite{XACT}, 
a {\em Wolfram Language} suite for doing tensor analysis (see also  \cite{XPERM}).

\section{Elements of conformal geometry}
\label{sec:conformal-eqs}
Let $(\tilde{\mathcal{M}},{\tilde g}_{ab})$ be a 4-dimensional Lorentzian manifold (referred to as the physical space-time) and
$(\mathcal{M},g_{ab})$ a se\-cond Lorentzian manifold (called the unphysical spacetime) which is {\em conformally related} to the
first in the following fashion (the signature convention for both metrics is $(-,+,+,+)$) 
\begin{equation}
g_{ab} = \Theta^2 \tilde{g}_{ab}.
\label{eq:unphysicaltophysical-downstairs}
\end{equation}
The previous relation can be used to define a conformal map (conformal embedding)
$\Phi:\tilde{\mathcal{M}}\rightarrow\mathcal{M}$ and the
conformal factor $\Theta$ is assumed to be a smooth function which does not vanish
on $\tilde{\mathcal{M}}$ (unless otherwise stated, all structures are assumed to be smooth).
Since $\Phi$ is an embedding, there exists an inverse map
$\Phi^{-1}: \Phi(\tilde{\mathcal{M}})\subset\mathcal{M}\rightarrow \mathcal{M}$.
The conformal embedding just defined is used to construct the \emph{conformal boundary}
of $(\tilde{\mathcal{M}},\tilde{g}_{ab})$ \cite{ZEROREST} and this is the reason for the
adopted terminology. It can also be used to set up an equivalence relation between diffeomorphic
spacetime manifolds under a conformal map. The resulting equivalence classes are called
\emph{conformal structures}. In this way, if \eqref{eq:unphysicaltophysical-downstairs} holds
then $g_{ab}$ and $\tilde{g}_{ab}$ define the same conformal structure.

We shall use small Latin letters to denote abstract indices
of tensors in $\mathcal{M}$ and $\tilde{\mathcal{M}}$. 
Indices are always raised and lowered with respect to the unphysical metric $g_{ab}$ 
with the exception of $\tilde{g}^{ab}$, where we follow
the traditional convention that it represents the inverse 
of $\tilde{g}_{ab}$
\begin{equation}
\tilde{g}^{ac}\tilde{g}_{cb}=\delta_b{}^a.
\end{equation}
From \eqref{eq:unphysicaltophysical-downstairs}, the explicit relation between
the physical and the unphysical contravariant metric tensors becomes
\begin{equation}
g^{ab} = \frac{\tilde{g}^{ab}}{\Theta^2}.
\label{eq:unphysicaltophysical-upstairs}
\end{equation}

Each of the metric tensors $g_{ab}$, $\tilde{g}_{ab}$ has its own volume element,
denoted respectively by $\eta_{abcd}$ and $\tilde{\eta}_{abcd}$.
Using again (\ref{eq:unphysicaltophysical-downstairs}) we deduce the relation
\begin{equation}
\tilde{\eta}_{abcd} = \frac{\eta_{abcd}}{\Theta^4}.
\label{eq:etaunphystoetaphys}
\end{equation}
Also each metric tensor has its own Levi-Civita connection denoted respectively 
by $\nabla_a$, $\tilde{\nabla}_a$ which are used to define the connection coefficients and the curvature 
tensors in the standard fashion. Our conventions for the (unphysical) Riemann, Ricci and Weyl tensors 
are
\begin{equation}
\nabla_{a}\nabla_{b}\omega_{c} -  \nabla_{b}\nabla_{a}\omega_{c} =  R_{abc}{}^{d} \omega_{d}\;,
\label{eq:define-riemann}
\end{equation}
\begin{equation}
 R_{ac}\equiv R_{abc}{}^b\;,
\label{eq:define-ricci}
\end{equation}
\begin{equation}
C_{abc}{}^{d}\equiv R_{abc}{}^{d} - 2 L^{d}{}_{[b}g_{a]c} - 2\delta^{d}{}_{[b}L_{a]c}\;,
\label{eq:define-weyl}
\end{equation}
where the unphysical Schouten tensor is defined by
\begin{equation}
 L_{ab} \equiv \tfrac{1}{2} (R_{ab} -  \tfrac{1}{6} R g_{ab}).
\label{eq:define-schouten}
\end{equation}
Tensors defined in terms of the physical metric $\tilde{g}_{ab}$ will be denoted with a tilde over
the symbol employed for an unphysical spacetime tensor.
The difference between the Levi-Civita connection
coefficients of the physical and the unphysical metrics
is the tensor $\Gamma[\nabla ,\tilde{\nabla}]^{e}{}_{ac}$ given by
\begin{equation}
\Gamma[\nabla ,\tilde{\nabla}]^{e}{}_{ac} = \delta_{c}{}^{e} \
\Upsilon_{a} -  g_{ca} \Upsilon^{e} + \delta_{a}{}^{e} \
\Upsilon_{c},\quad
\Upsilon_a\equiv\frac{\nabla_a\Theta}{\Theta}.
\label{eq:define-upsilon}
\end{equation}
Using the previous equation, standard computations enable us to find the relations between the curvature
tensors of $\nabla$ and $\tilde{\nabla}$.
In this work, the relation between the unphysical
Weyl tensor and the physical one will be specially important
\begin{equation}
C_{abc}{}^d=\tilde{C}_{abc}{}^d,\quad
\label{eq:weylunphystoweylphys}
\end{equation}
Here, $\tilde{C}_{abc}{}^{d}$ is the physical Weyl tensor as defined by \eqref{eq:define-weyl}
in terms of $\tilde{L}_{ab}$, $\tilde{g}_{de}$ and $\tilde{R}_{abc}{}^{d}$.
The star $*$ is used to denote both the Hodge dual (with respect to the physical or unphysical metric)
and the complex conjugation and we leave to the context the distinction between these.

The notions of \emph{metric concomitant} and \emph{conformally covariant} metric concomitant
are of fundamental relevance for this work.

\begin{definition}[Metric concomitant]
A \emph{metric concomitant} is a quantity defined in terms of the metric tensor,
its inverse and the Levi-Civita connection of the metric tensor.
\end{definition}

\begin{definition}[Conformally covariant metric concomitant]
 A \emph{metric concomitant} $T(g,\nabla)$ is said to be
\emph{conformally covariant} if it fulfills the following
relation under the transformation $g_{ab} = \Theta^2 \tilde{g}_{ab}$
\begin{equation}
{\tilde T}({\tilde g}_{ab}, \tilde\nabla)= \Theta^p T(g_{ab},\nabla),\quad
 p\in\mathbb{R},
 \label{eq:conformal-covariance}
\end{equation}
where the real number $p$ is called the \emph{conformal weight} of $T$
and ${\tilde g}_{ab}$, $g_{ab}$ are related by \eqref{eq:unphysicaltophysical-downstairs}
\end{definition}
For example the
Weyl tensor with the index configuration
as described by eq. \eqref{eq:weylunphys-cov-to-weylphys}
is conformally covariant with weight zero
(aka conformally invariant).
Any metric dependent quantity satisfying \eqref{eq:conformal-covariance}
can be regarded as a section of a suitable vector bundle of
\emph{conformal densities} (see \cite{CURRYGOVER}
for further details).

\section{An ideal conformally invariant characterization of Einstein spacetimes}
\label{sec:conformal-einstein}
Recall that a physical spacetime is Einstein iff
\begin{equation}
\tilde{R}_{ab}=\lambda \tilde{g}_{ab},\quad \lambda\in\mathbb{R}.
\label{eq:vacuum-lambda}
\end{equation}
A basic consequence of the Einstein field equations is that a vacuum solution is represented by an Einstein spacetime. The question now is whether for a given unphysical spacetime
$(\mathcal{M},g_{ab})$ we can find a \emph{local},
\emph{ideal} and \emph{conformally covariant} set of conditions that guarantee that $g_{ab}$ is conformal to an Einstein physical spacetime.
By local we mean that the conditions should hold in the neighbourhood of a point of $\mathcal M$ and by ideal conformally covariant
we mean that the conditions only involve the metric $g_{ab}$, its \emph{covariant concomitants} (tensors constructed out of the metric
and its Levi-Civita compatible covariant derivative) and also each of the conditions should
fulfill \eqref{eq:conformal-covariance} for some $p\in\mathbb{R}$.
This relevant question has been studied in \cite{KoNewTod85,Listing2001,Edgar2004,GOVER2006450}
The following Theorem is one of the results presented in \cite{GOVER2006450}
although we choose to present it in a form adapted to our purposes.

\begin{theorem}[Gover \& Nurowski 2006]
 Let $(\mathcal{M},g_{ab})$ be a 4-dimensional spacetime and define the following concomitants of $g_{ab}$.
 Define the unphysical metric concomitants
\begin{equation*}
\Lambda^d\equiv\left(\frac{8}{C\cdot C}\right)C^{damp}\nabla_{[m}L_{p]a},\quad
C\cdot C\equiv C_{abcd}C^{abcd}.
\end{equation*}
Then, assuming that $C\cdot C\neq 0$, $(\mathcal{M},g_{ab})$ is locally conformal to an Einstein spacetime if and only if the following
covariant concomitant of $g_{ab}$ vanishes
\begin{equation}
 E_{ab}(g_{ab},\nabla)\equiv\mbox{\rm TF}_{g}\big[
	L_{ab}+
	\nabla_{a}\Lambda_{b}
	+ \Lambda_a\Lambda_b\big],
\label{eq:ideal-einstein}
\end{equation}
where ${\rm TF}_g$ stands for ``trace-free part'' with respect to $g_{ab}$.
Moreover $E_{ab}$ is \emph{conformally invariant} according to
\begin{equation*}
 \tilde{E}_{ab}(\tilde{g}_{ab},\tilde\nabla)= E_{ab}(g_{ab},\nabla).
\end{equation*}
\label{theo:ideal-einstein}
\end{theorem}

\begin{proof}
The relation between the
unphysical and physical Schouten tensors can be always written in the form (see e.g. \cite{JUANBOOK})
\begin{equation}
L_{ab} = \tilde{L}_{ab} -  \Upsilon_{a} \Upsilon_{b} + \frac{1}{2} g_{ab} \Upsilon_{c} \Upsilon^{c} -  \nabla_{b}\Upsilon_{a}.
\label{eq:nabla-upsilon}
\end{equation}
Computing the unphysical covariant derivative of the previous equation and using \eqref{eq:define-riemann} with $\omega_d$ replaced by
$\Upsilon_d$ we get the following integrability condition
\begin{equation}
2\tilde{L}_{b[p} g_{m]a} \Upsilon^{b}
+2\tilde{L}_{a[m}\Upsilon_{p]} +
\Upsilon^{b} C_{abmp} + 2\nabla_{[p}\tilde{L}_{m]a} +
2\nabla_{[m}L_{p]a} = 0.
\label{eq:integrability-schouten}
\end{equation}
Assume now that we are under the hypotheses of
Theorem \ref{theo:ideal-einstein} and the physical spacetime is Einstein. Thus, according to
\eqref{eq:vacuum-lambda} we have $\tilde{L}_{ab}=\lambda\tilde{g}_{ab}/6$. Replacing this into
\eqref{eq:integrability-schouten} we obtain, after some manipulations involving
\eqref{eq:unphysicaltophysical-downstairs} and the definition of $\Upsilon_a$
\begin{equation}
\Upsilon^{b} C_{abmp} + 2\nabla_{[m}L_{p]a} = 0.
\end{equation}
If we multiply both sides of this equation by $C^{damp}$ and use the identity valid only
in dimension four,
$C_{abmp}C^{dbmp}=\tfrac{1}{4}\delta_a{}^dC_{cbmp}C^{cbmp}$
(see e.g. \cite{EDGARHOGLUNDDDI}),
it yields
\begin{equation}
 \Upsilon^d=\left(\frac{8}{C\cdot C}\right)C^{damp}\nabla_{[m}L_{p]a},
\label{eq:ideal-upsilon}
\end{equation}
where we assumed $C\cdot C\neq 0$. Replacing $\Upsilon^d$
back in \eqref{eq:nabla-upsilon}
we get $E_{ab} =0$. Conversely, assume that $E_{ab} =0$. Then from \eqref{eq:ideal-einstein}
we get
$$
\nabla_{[a}\Lambda_{b]} = 0.
$$
This means that locally $\Lambda_a=\nabla_a\log\psi$ for some function $\psi>0$.
Define now the physical metric ${\tilde g}_{ab}=\tfrac{g_{ab}}{\psi^2}$.
Then eq. \eqref{eq:nabla-upsilon} becomes
\begin{equation}
L_{ab} = \tilde{L}_{ab} -  \Lambda_{a} \Lambda_{b}
+ \frac{1}{2} g_{ab} \Lambda_{c} \Lambda^{c} -  \nabla_{b}\Lambda_{a}
\end{equation}
Replacing this value of $L_{ab}$ into \eqref{eq:ideal-einstein} we
conclude that ${\rm TF}_g(\tilde{L}_{ab})=0$ so $\tilde{L}_{ab}$ is proportional
to ${\tilde g}_{ab}$ and thus $\tilde{g}_{ab}$ is necessarily Einstein.
Finally, to prove the conformal invariance of $E_{ab}$ let us
assume the relation \eqref{eq:unphysicaltophysical-downstairs}
but this time we do not impose that the physical $\tilde{g}_{ab}$
is Einstein. Then we have the relations
(see eqs. \eqref{eq:Lambda-conformal}, \eqref{eq:nabla-upsilon} and \eqref{eq:define-upsilon})
\begin{eqnarray}
&&
  \tilde{g}_{de}\tilde{\Lambda}^{e} = \Lambda_{d} - \Upsilon_{d},\\
&&
  \tilde{L}_{ab} = L_{ab} + \Upsilon_{a} \Upsilon_{b} -  \tfrac{1}{2} \
  g_{ab} \Upsilon_{c} \Upsilon^{c} + \nabla_{b}\Upsilon_{a},\\
&&
  \tilde{\nabla}_{b}(\tilde{g}_{ae}\tilde\Lambda^{e}) = \Lambda_{b} \Upsilon_{a} + \
  \Lambda_{a} \Upsilon_{b} - 2 \Upsilon_{a} \Upsilon_{b} +
  g_{ab} (- \Lambda^{c} \Upsilon_{c} + \Upsilon_{c} \Upsilon^{c}) +
  \nabla_{b}\Lambda_{a} -  \nabla_{b}\Upsilon_{a}.
 \end{eqnarray}
Putting all together we get
\begin{equation}
\mbox{\rm TF}_{\tilde{g}}\big[
	\tilde{L}_{ab}+
	\tilde{\nabla}_{a}(\tilde{g}_{be}\Lambda^{e})
	+ \tilde{\Lambda}_a\tilde{\Lambda}_b\big] =
\mbox{\rm TF}_{g}\big[
	L_{ab}+
	\nabla_{a}\Lambda_{b}
	+ \Lambda_a\Lambda_b\big],
\end{equation}
which is the sought conformal invariance.
\end{proof}

Eq. \eqref{eq:ideal-upsilon} is an useful relation that enables us to write the differential of $\Theta$ in terms of concomitants
of the metric $g_{ab}$ whenever $\Theta$ defines a conformal relation between an unphysical metric and
a physical Einstein space in dimension 4. See \cite{Edgar2004} to find out how the previous computation can be generalized to any dimension.

If the unphysical spacetime does not fulfill the hypotheses of Theorem \ref{theo:ideal-einstein} then \eqref{eq:ideal-upsilon}
does not necessarily hold. However, we can still define a concomitant of the unphysical metric $g_{ab}$ using the right hand side
of this equation
\begin{equation}
\Lambda^d\equiv\left(\frac{8}{C\cdot C}\right)C^{damp}\nabla_{[m}L_{p]a}.
\label{eq:ideal-lambda}
\end{equation}
The concomitant $\Lambda^d$ will be needed in our main result, Theorem \ref{theo:kerr-conformal}. Clearly
if $\tilde{g}_{ab}$ is Einstein then $\Lambda_a=\Upsilon_a$.

\section{An ideal characterization of the physical Kerr solution}
\label{sec:type-D}
A set of \emph{local} characterizations of the Kerr solution
were presented in \cite{MARS-KERR-UNIQUENESS, MARS-KERR}. However,
these results are not ideal in the sense that they necessitate the
existence of a Killing vector field (in \cite{MARS-KERR-UNIQUENESS} also
the asymptotic flatness is required). In \cite{FERSAEZKERR} a local ideal characterization of
the Kerr solution was produced. This characterization of the Kerr solution
will be taken as our starting point for the ideal characterization of
the Kerr conformal family, so it will be next reviewed.
Let us introduce the required geometric elements:
from the physical metric $\tilde{g}_{ab}$ and
its inverse $\tilde{g}^{ab}$, we define its volume element $\tilde{\eta}_{abcd}$,
the physical Weyl tensor $\tilde{W}_{abcd}\equiv\tilde{g}_{ah}\tilde{C}_{bcd}{}^h$,
the physical Weyl tensor right dual $\tilde{W}^*_{abcd}$ with respect to the physical metric
and the physical self-dual Weyl tensor
\begin{equation}
\tilde{\mathcal{W}}_{abcd}\equiv\frac{1}{2} (\tilde{W}_{abcd} - {\rm i} \tilde{W}^{*}{}_{abcd}),\quad
\tilde{W}^{*}{}_{abcd}\equiv\frac{1}{2}\tilde{\eta}_{cdpq}\tilde{g}^{pr}\tilde{g}^{qs}\tilde{W}_{abrs}.
\label{eq:define-weyl-self-dual}
\end{equation}
Next we introduce the following physical Weyl scalars
\begin{eqnarray}
&&\tilde{\mathit{a}} \equiv \tilde{g}^{ac}\tilde{g}^{be}\tilde{g}^{fq}\tilde{g}^{pd} \tilde{\mathcal{W}}_{abpf} \tilde{\mathcal{W}}_{cedq}\;,\\
&&\tilde{\mathit{b}} \equiv \tilde{g}^{ai}\tilde{g}^{bj}\tilde{g}^{ce}\tilde{g}^{dg}\tilde{g}^{pf}\tilde{g}^{qh}\tilde{\mathcal{W}}_{abcd} \tilde{\mathcal{W}}_{egpq} 
\tilde{\mathcal{W}}_{fhij}\;,\\
&&\tilde{\mathit{w}}\equiv-\frac{\tilde{\mathit{b}}}{2\tilde{\mathit{a}}}\;,\quad \tilde{a}\neq 0
\end{eqnarray}
and the tensors
\begin{equation}
\tilde{G}_{abmh} \equiv \tilde{g}_{am} \tilde{g}_{bh} -  \tilde{g}_{ah} \tilde{g}_{bm}\;,\quad
\tilde{\mathcal{G}}_{abcd} \equiv \tfrac{1}{2} (-{\rm i} \tilde{\eta}_{abcd} + \tilde{G}_{abcd}).
\end{equation}

The Kerr solution belongs to the class of vacuum \emph{type D} solutions.  An
ideal characterization of the latter family of algebraically special solutions was
presented in  \cite{FERRANDOCOVARIANTPETROVTYPE,FERRSAEZTYPED}. We reproduce here this
ideal characterization, adopting the formulation of \cite{AGPTYPEDDATA}.
\begin{theorem}
The physical spacetime $(\tilde{\mathcal{M}},\tilde{g}_{ab})$ is of ``genuine'' Petrov type D (Petrov type D, but not any of its specializations) if and only if 
\begin{equation}
\tilde{a}\neq 0\;,\quad \tilde{\mathcal{D}}_{abcd}=0\;,
\label{eq:type-D}
\end{equation}
where
\begin{equation}
\tilde{\mathcal{D}}_{abhe}\equiv \tilde{\mathcal{W}}_{abcd}\tilde{g}^{cp} 
\tilde{g}^{dq}\tilde{\mathcal{W}}_{pqhe} -  
\frac{\tilde{\mathit{a}}}{6}\tilde{\mathcal{G}}_{abhe} - 
\frac{\tilde{\mathit{b}}}{\tilde{\mathit{a}}}\tilde{\mathcal{W}}_{abhe}.
\end{equation}
\label{theo:type-D}
\end{theorem}

The following result, proven in \cite{FERSAEZKERR}, is the ideal characterization of
the \emph{physical Kerr family} we are going to use.
The formulation presented here is taken from
\cite{AGPTYPEDDATA,AGPTYPEDDATA-ERRATUM}.
\begin{theorem}[Ideal characterization of the Kerr solution]
Under the conditions of Theorem \ref{theo:type-D}
a vacuum (Ricci flat $\tilde{R}_{ab}=0$) space-time $(\tilde{\mathcal{M}},\tilde{g}_{ab})$ is locally isometric to the Kerr
solution with non-vanishing mass (non-trivial Kerr solution) if and only if the following additional conditions hold
\begin{eqnarray}
&&\tilde{\Xi}_{a[b}\tilde{\Xi}^*_{c]d}=0\;,
\label{eq:killing-property-kerr}\\
&&\mbox{\em Im}(\tilde{Z}^{3}(\tilde{w}^*)^8)=0\;,
\label{eq:nutzero}\\
&&
\left\{
\begin{array}{l}
	\frac{\mbox{\em Re}(\tilde{Z}^{3}(\tilde{w}^*)^8)}{\big(18 \mbox{\em Re}\big(\tilde{w}^3\tilde{Z}^*\big)-|\tilde{Z}|^2\big)^3}<0,\;
	(\mbox{if}\ 18 \mbox{\em Re}\big(\tilde{w}^3\tilde{Z}^*\big)-|\tilde{Z}|^2\neq 0),\;\\
	\\
	\mbox{\em Re}(\tilde{Z}^{3}(\tilde{w}^*)^8)=0,\;
	(\mbox{if}\ 18 \mbox{\em Re}\big(\tilde{w}^3\tilde{Z}^*\big)-|\tilde{Z}|^2= 0)\;,
\end{array}
\right.
\label{eq:epsilong0}
\end{eqnarray}
where $\tilde{\Xi}_{ab}$, $\tilde{Z}$ are the concomitants of $\tilde{g}_{ab}$  defined by
\begin{eqnarray}
	&&\tilde{\Xi}_{ac} \equiv (
		\tilde{g}^{bp} \tilde{g}^{dq} \tilde{\mathcal{W}}_{apcq} -
		\tilde{\mathit{w}}\tilde{\mathcal{G}}_{apcq}
		\tilde{g}^{bp}
		\tilde{g}^{dq}
	)
	\tilde{\nabla}_{b}\tilde{\mathit{w}}\tilde{\nabla}_{d}\tilde{\mathit{w}},\label{eq:define-tildexi}\\
	&& \tilde{Z}\equiv \tilde{g}^{ab}\tilde{\nabla}_a \tilde{w}\tilde{\nabla}_b\tilde{w}.\label{eq:define-tildez}
\end{eqnarray}

\label{theo:kerr-local}
\end{theorem}

\begin{remark}\label{rem:kerr-nut}\em
 Theorem \ref{theo:kerr-local} can be also used as a
local characterization of the Kerr-NUT solution if
we only consider the condition given by \eqref{eq:killing-property-kerr}.
See \cite{AGPTYPEDDATA,AGPTYPEDDATA-ERRATUM} for further details.

\end{remark}

\section{A conformally covariant ideal characterization of the
conformal Kerr family}
\label{sec:kerr-conformal}

The next step is to write the results contained in Theorems
\ref{theo:type-D} and \ref{theo:kerr-local} in terms of concomitants of the unphysical metric $g_{ab}$.
To find the corresponding formulations we need to make similar definitions for the symbols used in these theorems, but now using the unphysical metric instead of the
physical one. The notation for the new symbols so defined is obtained by just removing the tildes over the symbols
used in the physical space-time.

\begin{lemma}
One has the following relation between the following physical quantities defined in Theorems \ref{theo:type-D}
and \ref{theo:kerr-local} and their unphysical counterparts

\begin{subequations}
\begin{eqnarray}
&&
C_{abcd} = \Theta^2 \tilde{W}_{abcd},
\label{eq:weylunphys-cov-to-weylphys}\\
&&\tilde{\mathit{a}} = \Theta^4 \mathit{a}\;,\\
&&\tilde{\mathit{b}} = \Theta^6 \mathit{b}\;,\\
&&\tilde{\mathit{w}} = \Theta^2 \mathit{w}\;,\label{eq:conformal-rescaling-w}\\
&&\widetilde{G}_{abcd} = \frac{G_{abcd}}{\Theta^4}\;,\\
&&\tilde{\mathcal{G}}_{abcd} = \frac{\mathcal{G}_{abcd}}{\Theta^4}\;,\\
&& \tilde{\mathcal{W}}_{abcd} =\frac{{\mathcal W}_{abcd}}{\Theta^2}\;,\\
&&\widetilde{\mathcal{D}}_{abcd}=\mathcal{D}_{abcd}\;,\label{eq:conformal-rescaling-D}
\end{eqnarray}
\end{subequations}
\label{lem:rescaling-properties}
\end{lemma}
\proof This is a straightforward computation carried out by using the relations
(\ref{eq:unphysicaltophysical-downstairs}), (\ref{eq:unphysicaltophysical-upstairs}),
(\ref{eq:etaunphystoetaphys}), (\ref{eq:weylunphystoweylphys})

\qed
\begin{remark}\em
The proof of Lemma \ref{lem:rescaling-properties} does not assume that the physical metric is Einstein.
Therefore equations \eqref{eq:weylunphys-cov-to-weylphys}-\eqref{eq:conformal-rescaling-D} tell us that
the involved metric concomitants are \emph{conformally covariant}.
\end{remark}
\noindent
If we look at the concomitants used in
Theorem \ref{theo:kerr-local}
(see eqns. \eqref{eq:define-tildexi}-\eqref{eq:define-tildez}) we realize that there
are concomitants defined in terms of $\tilde\nabla_a\tilde w$ that is not conformally covariant
under \eqref{eq:unphysicaltophysical-downstairs}
as can be explicitly checked
\begin{equation}
 \tilde\nabla_a\tilde w= 2 w\Upsilon_a+\nabla_a w.
\end{equation}
If the physical space-time is Einstein, then we can use in the previous equation the relation
$\Upsilon_a=\Lambda_a$ where $\Lambda_a$ is the concomitant of the unphysical
metric defined by \eqref{eq:ideal-lambda} (see also \eqref{eq:ideal-upsilon}). This leads us
to the definition of a new concomitant of the unphysical metric $g_{ab}$, be it conformal to an Einstein space or not
\begin{equation}
 \lambda_a\equiv 2 w\Lambda_a+\nabla_a w.
\label{eq:define-lambda}
\end{equation}
\begin{lemma}
Under the transformation \eqref{eq:unphysicaltophysical-downstairs}, one has the relation
\begin{equation}
\tilde\lambda_a =\Theta^2 \lambda_a,
\label{eq:lambda-conf-transf}
\end{equation}
where as usual $\tilde\lambda_a$ is defined as in \eqref{eq:define-lambda} replacing the
unphysical by the physical metric.
\label{lemm:rescale-lambda}
\end{lemma}
\proof
Expressing $\nabla$ in terms of $\tilde\nabla$ by means of \eqref{eq:define-upsilon}
we obtain the following equivalent form of \eqref{eq:integrability-schouten}
\begin{equation}
\tilde{\nabla}_{[b}\tilde{L}_{e]a} = \tfrac{1}{2} \Upsilon^{c} C_{acbe} + \nabla_{[b}L_{e]a}.
\end{equation}
Using this relation and \eqref{eq:weylunphys-cov-to-weylphys} in \eqref{eq:ideal-lambda}
we get
\begin{equation}
 \tilde{\Lambda}^{d} = \Theta^2(\Lambda^{d} - \Upsilon^{d}).
 \label{eq:Lambda-conformal}
\end{equation}
Combining this with the definition of $\Upsilon^d$ and \eqref{eq:conformal-rescaling-w}
we can obtain \eqref{eq:lambda-conf-transf}.
\qed

\begin{lemma}
Under the transformation \eqref{eq:unphysicaltophysical-downstairs},
and assuming that the physical metric $\tilde{g}_{ab}$ is Einstein one has the relation
\begin{equation}
\tilde{\Xi}_{ac} = \Theta^4 \xi_{ac}\;,\label{eq:conformal-rescaling-tildeXi}
\end{equation}
where we have defined
\begin{equation}
\xi_{ac}\equiv
(\mathcal{W}_{a}{}^{b}{}_{c}{}^{d}-\mathit{w}\mathcal{G}_{a}{}^{b}{}_{c}{}^{d})
\lambda_b\lambda_d,\quad
\lambda_a\equiv 2 w\Lambda_a+\nabla_a w\;,
\label{eq:xitoxi0}
\end{equation}

\end{lemma}
\proof
If we use the transformation \eqref{eq:unphysicaltophysical-downstairs}
on the definition of $\tilde\Xi_{ab}$ (given by \eqref{eq:define-tildexi})
and the results of Lemma \ref{lem:rescaling-properties}
then it is very easy to obtain a relation similar to \eqref{eq:xitoxi0}
with $\Lambda_a$ replaced by $\Upsilon_a$.
Now, if the physical space-time is Einstein, then
the relation $\Upsilon_a=\Lambda_a$ leads to the desired result
(see the considerations after eq. \eqref{eq:ideal-lambda}).
\qed

\begin{remark}\em
The tensor $\xi_{ac}$ is a concomitant of the unphysical metric $g_{ab}$
regardless of whether it is conformal to an Einstein metric. In fact
$\xi_{ac}$ is conformally covariant as it is deduced from the results of
Lemmas \ref{lem:rescaling-properties} and \ref{lemm:rescale-lambda}.
Using these results we deduce that under the transformation \ref{eq:unphysicaltophysical-downstairs}
we have the relation
\begin{equation}
 \tilde{\xi}_{ab} = \Theta^6 \xi_{ab}.
\label{eq:rescale-xi}
\end{equation}

\end{remark}

\begin{theorem}
The unphysical spacetime $(\mathcal{M},g_{ab})$ is locally conformal to a Petrov type D physical
spacetime if and only if $\mathcal{D}_{abcd}=0$ on $\mathcal{M}$.
\label{theo:typeDunphysical}
\end{theorem}
\proof
This is a direct consequence of eq.(\ref{eq:conformal-rescaling-D}) and
Theorem \ref{theo:type-D}.
\qed

\medskip
We are now ready to formulate our main result.
\begin{theorem}[Ideal characterization of the conformal Kerr structure]
There exists an open subset of the unphysical spacetime $(\mathcal{M},g_{ab})$ that is locally conformal
to the non-trivial Kerr solution if and only if
the conditions of Theorem \ref{theo:ideal-einstein}
together with the following conditions are met in a neighbourhood of a point of $\mathcal{M}$
\begin{subequations}
\begin{eqnarray}
&&a\neq 0\;,\quad \mathcal{D}_{abcd}=0\;,\quad
\label{eq:typeD-property-cfkerr}\\
&& \xi_{a[b}\xi^*_{c]d}=0\;,\quad
\label{eq:killing-property-cfkerr}\\
&&\mbox{\em Im}\big(z^3(w^*)^8\big)=0\;,
\label{eq:nut-unphysical-zero}\\
&&
\left\{
\begin{array}{ll}
\frac{\mbox{\em Re}\big(z^3(w^*)^8\big)}{\bigg(18\mbox{\em Re}(w^3 z^*)-|z|^2\bigg)^3}<0\;, &
\mbox{if}\ 18\mbox{\em Re}(w^3 z^*)-|z|^2\neq 0\;,\\
\mbox{\em Re}(z^3(w^*)^8 )=0\;,\quad  & \mbox{if}\ 18\mbox{\em Re}(w^3 z^*)-|z|^2=0,
\end{array}
\right.
\label{eq:unphysical-epsilong0}
\end{eqnarray}
where
for the sake of readability we give again the definitions of some of the involved concomitants of the unphysical metric
(the definition of $\mathcal{W}_{abcd}$ is Eq. \eqref{eq:define-weyl-self-dual} with the tildes removed)
\begin{eqnarray}
&&
a\equiv\mathcal{W}_{abcd}\mathcal{W}^{abcd},\quad
b\equiv\mathcal{W}_{abcd}\mathcal{W}^{cd}{}_{he}\mathcal{W}^{heab},\quad
w\equiv-\frac{b}{2a},
\\
&&
\mathcal{D}_{abhe}\equiv \mathcal{W}_{abcd}\mathcal{W}^{cd}{}_{he} -
\frac{\mathit{a}}{6}\mathcal{G}_{abhe} -
\frac{\mathit{b}}{\mathit{a}}\mathcal{W}_{abhe}\;,\quad
\Lambda_d\equiv\left(\frac{8}{C\cdot C}\right)C_{d}{}^{amp}\nabla_{[m}L_{p]a},\\
&&
\xi_{ac}\equiv
(\mathcal{W}_{a}{}^{b}{}_{c}{}^{d}-\mathit{w}\mathcal{G}_{a}{}^{b}{}_{c}{}^{d})
\lambda_b\lambda_d,\quad
\lambda_a\equiv 2 w\Lambda_a+\nabla_a w\;,\quad
z\equiv g^{ab}\lambda_a\lambda_b\;.
\label{eq:define-Z}
\end{eqnarray}
\label{eqns:conformal-kerr}
\end{subequations}
\label{theo:kerr-conformal}
\end{theorem}

\proof
As stated by Theorem \ref{theo:ideal-einstein} the unphysical spacetime
is conformal to an Einstein spacetime if and only if
\eqref{eq:ideal-einstein} holds. Therefore for the remainder of the proof
we may assume that the unphysical spacetime metric tensor $g_{ab}$ is conformally related
to an Einstein physical spacetime metric tensor $\tilde{g}_{ab}$ by means of
\eqref{eq:unphysicaltophysical-downstairs} and hence all the
formulae derived from this equation hold.
Also since $\tilde{g}_{ab}$ is Einstein, we have the relation $\Lambda_a=\Upsilon_a$
(see the considerations after \eqref{eq:ideal-lambda})
Thus, using \eqref{eq:conformal-rescaling-w} found in Lemma \ref{lem:rescaling-properties}
we obtain the relation
\begin{equation}
\tilde{Z}=\Theta^6 z,
\label{eq:conformal-rescaling-Z}
\end{equation}
with $z$ as defined by \eqref{eq:define-Z} and $\tilde{Z}$ as defined by \eqref{eq:define-tildez}.
From \eqref{eq:conformal-rescaling-Z} we deduce, using again \eqref{eq:conformal-rescaling-w} and
\eqref{eq:conformal-rescaling-Z}
\begin{eqnarray}
&&
\mbox{Re}(\tilde{Z}^{3}(\tilde{w}^*)^8)=\Theta^{34}\mbox{Re}\big(z^3(w^*)^8\big),\quad
\mbox{Im}(\tilde{Z}^{3}(\tilde{w}^*)^8)=\Theta^{34}\mbox{Im}\big(z^3(w^*)^8\big),\label{eq:conformal-rescaling-nutzero}\\
&&
18 \mbox{Re}\big(\tilde{w}^3\tilde{Z}^*\big)-|\tilde{Z}|^2 = \Theta^{12} (18\mbox{Re}(w^3 z^*)-|z|^2),
\end{eqnarray}
hence
\begin{equation}
 \frac{\mbox{Re}(\tilde{Z}^{3}(\tilde{w}^*)^8)}
 {\big(18 \mbox{Re}\big(\tilde{w}^3\tilde{Z}^*\big)-|\tilde{Z}|^2\big)^3} =
\frac{1}{\Theta^2}\frac{\mbox{Re}\big(z^3(w^*)^8\big)}{\bigg(18\mbox{Re}(w^3 z^*)-|z|^2\bigg)^3}.
\label{eq:conformal-rescaling-epsilong0}
\end{equation}
Let us assume now that $\tilde{g}_{ab}$ is the Kerr solution. Then \eqref{eq:type-D} holds due to Theorem
\ref{theo:type-D} and (\ref{eq:killing-property-kerr})-(\ref{eq:epsilong0}) also hold due
to Theorem \ref{theo:kerr-local}. Using the results of Lemma \ref{lem:rescaling-properties},
\eqref{eq:conformal-rescaling-nutzero}-\eqref{eq:conformal-rescaling-epsilong0} and the relation $\Lambda_a=\Upsilon_a$
one obtains \eqref{eq:typeD-property-cfkerr}-\eqref{eq:nut-unphysical-zero} and \eqref{eq:unphysical-epsilong0}.

Conversely, let us assume that \eqref{eq:typeD-property-cfkerr}-\eqref{eq:unphysical-epsilong0} are true.
Then the relations of Lemma \ref{lem:rescaling-properties},
\eqref{eq:conformal-rescaling-nutzero}-\eqref{eq:conformal-rescaling-epsilong0} and $\Lambda_a=\Upsilon_a$ (valid because the
physical metric is Einstein) entail \eqref{eq:type-D}, (\ref{eq:killing-property-kerr})-(\ref{eq:epsilong0})
from which we conclude that the physical metric is the Kerr solution due to Theorem \ref{theo:kerr-local}.\qed

\begin{remark}\em
All the quantities involved in \eqref{eqns:conformal-kerr} are conformally
covariant, as proven by Lemmas \ref{lem:rescaling-properties}, \ref{lemm:rescale-lambda} and eq. \eqref{eq:rescale-xi}.
Therefore the characterization of the Kerr conformal family presented by
Theorem \ref{theo:kerr-conformal} is ideal and conformally covariant.
\label{rem:conformal-covariance}
\end{remark}

\begin{remark}\em
From the considerations of Remark \ref{rem:kerr-nut} we deduce that
Theorem \ref{theo:kerr-conformal} also yields an ideal characterization of
the conformal family defined by the Kerr-NUT solution if we just keep
the conditions \eqref{eq:typeD-property-cfkerr}-\eqref{eq:killing-property-cfkerr}.
Moreover, it only \eqref{eq:typeD-property-cfkerr} is kept we
obtain an ideal characterization of the conformal family of
(vacuum) type D solutions.
\end{remark}

\section{Conclusions}
\label{sec:conclusions}
In the proof of Theorem \ref{theo:kerr-conformal} the unphysical spacetime $(\mathcal{M},g_{ab})$
is conformal to a physical spacetime $(\tilde{\mathcal{M}},\tilde{g}_{ab})$ that is Einstein, since
Theorem \ref{theo:ideal-einstein} is assumed to hold.
This means that a conformal embedding $\Phi:\tilde{\mathcal{M}}\rightarrow\mathcal{M}$ exists and hence
Theorem \ref{theo:kerr-conformal} only applies to an open subset $\mathcal{U}\subset\mathcal{M}$ such that $\Phi^{-1}(\mathcal{U})\subset\tilde{\mathcal{M}}$.
However, if all concomitants displayed in \eqref{eq:typeD-property-cfkerr}-\eqref{eq:unphysical-epsilong0} are smooth in
$\mathcal{M}$ and the conditions of Theorem \ref{theo:ideal-einstein} hold only in an open subset $\mathcal{U}\subset\mathcal{M}$,
then it is plausible to think that the unphysical spacetime $(\mathcal{M}, g_{ab})$ is a \emph{conformal extension} of the Kerr physical spacetime $(\tilde{\mathcal{M}}, \tilde{g}_{ab})$.
In this sense the result presented in this article would be the first step towards an ideal characterization of
a smooth conformal extension of the Kerr solution which in turn could provide a construction of the \emph{Kerr conformal boundary}.
It is worth noting that studies about initial data on the conformal boundary of the Kerr solution
are already available in the literature when the cosmological
constant does not vanish and the conformal boundary is spacelike
(namely, the Kerr de Sitter solution, see  \cite{Mars2021, Mars_2017, MSENOTORB16, MR4392180, Mars_2021}). Our result could be used for a similar purpose
when the cosmological constant vanishes if
one finds suitable \emph{hyperbolic reductions} and (characteristic) initial
data for the conditions found in Theorem \ref{theo:kerr-conformal}.

As already pointed out in Remark \ref{rem:conformal-covariance},
looking at the covariant concomitants of the unphysical metric given by
Theorem \ref{theo:kerr-conformal} we realize that all of them are \emph{conformally covariant}.
This renders our characterization of the Kerr conformal family as conformally covariant,
a property already shared by family of 4-dimensional metrics that are conformal
to the vacuum solutions of the Einstein field equations (see Theorem \ref{theo:ideal-einstein}).
To the best of our knowledge this is
a new result about the Kerr (conformal) family whose consequences are still to be investigated.

\section*{Acknowledgements}
We thank the Department of Astronomy and Astrophysics of València University in Spain
where part of this work was carried out for hospitality. We also thank an anonymous referee
for inte\-resting comments.

\noindent
Supported by Spanish MICINN Project No. PID2021-126217NB-I00.

% \bibliographystyle{amsplain}
% \bibliography{/home/alfonso/trabajos/BibDataBase/Bibliography}

% \end{document}

\providecommand{\bysame}{\leavevmode\hbox to3em{\hrulefill}\thinspace}
\providecommand{\MR}{\relax\ifhmode\unskip\space\fi MR }
\providecommand{\MRhref}[2]{%
  \href{http://www.ams.org/mathscinet-getitem?mr=#1}{#2}
}
\providecommand{\href}[2]{#2}

\end{document}